\documentclass[twocolumn,letter]{IEEEtran}
\usepackage{latexsym}
\usepackage{amssymb}
\usepackage{amsmath}

\usepackage{color}
\usepackage{bm}
\usepackage[dvips]{graphicx}

\newtheorem{thm}    {Theorem}
\newtheorem{lem}     {Lemma}
\newtheorem{condition}  {Condition}

\newcommand{\red}{}

\def\mix{\mathop{\rm mix}}

\newcommand{\defeq}{\stackrel{\rm def}{=}}

\newcommand{\bF}{\mathbb{F}}

\def\cA{{\cal A}}
\def\cE{{\cal E}}
\def\cY{{\cal Y}}
\def\cZ{{\cal Z}}

\def\cD{{\cal D}}

\def\rE{{\rm E}}
\def\rP{{\rm P}}

\newcommand{\cX}{{\cal X}}

\newcommand{\bX}{{\bf X}}
\newcommand{\bY}{{\bf Y}}

\makeatletter
\newcommand{\lleq}{\mathrel{\mathpalette\gl@align<}}
\newcommand{\ggeq}{\mathrel{\mathpalette\gl@align>}}
\newcommand{\gl@align}[2]{
\vbox{\baselineskip\z@skip\lineskip\z@
\ialign{$\m@th#1\hfil##\hfil$\crcr#2\crcr{}_{{}_{(=)}}\crcr}}}
\makeatother

\def\Label#1{\label{#1}\ [\ #1\ ]\ }
\def\Label{\label}

\begin{document}
\title{Exponential decreasing rate of leaked information\\ in universal random privacy amplification}
\author{
Masahito Hayashi
\thanks{
M. Hayashi is with Graduate School of Information Sciences, Tohoku University, Aoba-ku, Sendai, 980-8579, Japan
(e-mail: hayashi@math.is.tohoku.ac.jp)}}
\date{}
\maketitle
\begin{abstract}
We derive
a new upper bound for Eve's information in secret key generation from a common random number without communication.
This bound improves on Bennett\cite{BBCM}'s bound based on the R\'{e}nyi entropy of order 2
because the bound obtained here uses the R\'{e}nyi entropy of order $1+s$ for $s \in [0,1]$.
This bound is applied to a wire-tap channel.
Then, we derive an exponential upper bound for Eve's information.
Our exponent is compared with Hayashi\cite{Hayashi}'s exponent.
For the additive case, the bound obtained here is better.
The result is applied to secret key agreement by public discussion.
\end{abstract}
\section{Introduction}
\PARstart{T}{he}
study of secure communication in the presence of an eavesdropper began with 
Wyner\cite{Wyner}.
Following Wyner, 
Csisz\'{a}r \& K\"{o}rner\cite{CK79} 
dealt with this topic.
In this study, we consider a sender Alice,
an authorized receiver Bob and an unauthorized receiver Eve, who is referred to as a wire-tapper.
This research treats two channels, a channel to Bob and a channel to Eve;
such a model is called a wire-tap channel. 
Whereas the studies above treated the discrete memoryless case,
Hayashi\cite{Hayashi} derived a general capacity formula for an
arbitrary sequence of wire-tap channels.
In this model,
amount of Eve's accessible information is evaluated by  
the mutual information $I_E (\Phi)$ between Alice's and Eve's variables with the code $\Phi$,
and is abbreviated to Eve's information.
Several papers \cite{Ren05,RW,Cannetti} in cryptography community
adopt the leaked information criterion
based on the variational distance
while
several papers \cite{AC93,CK79,Hayashi,Wyner,Csiszar,Naito} in information theory community
adopt the leaked information criterion
based on the mutual information.
As is illustrated in Appendix III,
there exists an example where 
the leaked information criterion based on the mutual information 
is more restrictive than that based on variational distance.
Hence, we adopt the leaked information criterion based on the mutual information.

As was shown by Csisz\'{a}r \cite{Csiszar},
in the discrete memoryless case, 
if the transmission rate is less than the capacity and if we choose suitable codes,
Eve's information goes to zero exponentially.
That is, when
the given channel is used with $n$ times,
Eve's information $I_E (\Phi_n)$ with a suitable code $\Phi_n$ behaves as $e^{-n r}$.
In order to estimate the speed of the convergence,
we focus on the {\it exponential decreasing rate} of Eve's information,
which is referred to as the {\it exponent} of Eve's information:
\begin{align}
\lim_{n \to \infty}
\frac{-1}{n}\log I_E (\Phi_n).
\end{align}
Hayashi\cite{Hayashi} estimates 
this exponent 
for the wire-tap channels in the discrete memoryless case.
This type of evaluation is quite useful for estimating Eve's information from a finite-length code.
The first purpose of this paper is to improve the previous 
exponent of Eve's information.

On the other hand, 
using the R\'{e}nyi entropy of order 2,
Bennett et al \cite{BBCM}
evaluate Eve's information after the application of a universal$_2$ hashing function\cite{Carter}.
Their result gives an upper bound of Eve's information for the generation of a secret key from a common random number without communication.
Renner and Wolf \cite{RW} and
Renner \cite{Ren05} improved this approach 
and obtained evaluations based on smooth R\'{e}nyi entropy.
Renner \cite{Ren05} applied his method to the security analysis of quantum key distribution.
However, no research studied the relation between these results related to various kinds of R\'{e}nyi entropies
and the above results concerning wire-tap channel.

The main purpose of this paper is to generalize Bennett et al \cite{BBCM}'s result and to apply it to wire-tap channel model.
As the first step, 
in Section \ref{s4}, 
we focus on secret key generation from a common random number without communication.
Even in this model,
we highlight the exponent of Eve's information
in the case of independent and identical distribution (i.i.d. case).
In subsection \ref{s41},
we extend the result of Bennett et al \cite{BBCM} 
to the case of the R\'{e}nyi entropy of order $1+s$ for $s \in [0,1]$
and obtain
a new upper bound for Eve's information in this problem
as the main theorem.
We apply this bound to the i.i.d. case.
Then, derived a lower bound of 
the exponent of Eve's information.
In subsection \ref{s42},
we also apply Renner and Wolf \cite{RW}'s method to the evaluation of 
the exponent of Eve's information.
Then, another lower bound is derived based on smooth R\'{e}nyi entropy.
It is shown that 
the lower bound based on R\'{e}nyi entropy of order $1+s$  
is better than that based on smooth R\'{e}nyi entropy.

In Section \ref{s2}, 
\red{based on universal$_2$ hash function,}
we derive an upper bound for Eve's information from random coding in a wire-tap channel.
The method we present contrasts
with the method in Hayashi\cite{Hayashi}.
Hayashi\cite{Hayashi} deals with channel resolvability and applies it to the security of 
wire-tap channel; This approach was strongly motivated by Devetak \cite{Deve} 
and Winter et al \cite{WNI}.
In Section \ref{s6}, we show that
this upper bound for Eve's information is better than 
Hayashi\cite{Hayashi}'s bound for the wire-tap channel model.

In a realistic setting, 
it is usual to restrict our codes to linear codes.
However, no existing result gives a code satisfying the following conditions:
(1) The code is constructed by linear codes.
(2) Eve's information exponentially goes to zero when 
the transmission rate is smaller than 
the difference between the mutual information from Alice to Bob and that to Eve.
In Section \ref{s3}, we make a code satisfying the above conditions.
That is, we make our code generated by a combination of arbitrary linear codes 
and privacy amplification by 
the concatenation of Toeplitz matrix \cite{Krawczyk} and the identity.
Under this kinds of code, 
\red{applying the evaluation obtained in subsection \ref{s41}
and the concavity property of the key quantity given in \ref{s2},}
we obtain \red{another} upper bound for Eve's information.
when the channel is an {\it additive} channel, i.e., 
the probability space and the set of input signals are given as 
the same finite module and 
the probability transition matrix $W_a(b)$ corresponding to the channel
is given as $P(a-b)$ with a probability distribution on the finite module.
This fact holds when
the channel is a variant of an additive channel.

In Section \ref{s5},
we also apply our result to secret key agreement with public discussion,
which has been treated by
Ahlswede \& Csisz\'{a}r\cite{AC93}, Maurer\cite{Mau93}, 
and Muramatsu\cite{MUW05} et al.
Maurer \cite{Mau93} and Ahlswede \& Csisz\'{a}r\cite{AC93} 
showed that the optimal 
key generation rate is the difference of conditional entropies $H(A|E)-H(A|B)$,
where $A$, $B$, $E$  are the random variables for 
Alice, Bob, and Eve, respectively.
Csisz\'{a}r\cite{Csiszar}, Renner\cite{Ren05}, and Naito et al \cite{Naito}
mentioned the existence of 
a bound for Eve's information that exponentially goes to zero when the key generation rate is smaller than $H(A|E)-H(A|B)$.
However, no existing result 
clearly gives a lower bound for the exponential decreasing rate
for Eve's information 
when the key generation rate is smaller than $H(A|E)-H(A|B)$.
Applying our result, 
we obtain such a lower bound for the exponential decreasing rate for Eve's information.
In this case, we apply our code to a wire-tap channel with a variant of additive channels.
Our protocol can be realized by 
a combination of
a linear code and privacy amplification by 
the concatenation of Toeplitz matrix \cite{Krawczyk} and the identity.

In Appendix \ref{s7}, we prove the main theorem mentioned in Section \ref{s4}.
In Appendix \ref{stoep}, we show that the concatenation of Toeplitz matrix \cite{Krawczyk} and the identity is a universal$_2$ hashing function \cite{Carter}.

\section{Secret key generation without communication}\Label{s4}
\subsection{Method based on R\'{e}nyi entropy of order $1+s$}\Label{s41}
Firstly,
we consider the secure key generation problem from
a common random number $a \in \cA$ which has been partially eavesdropped on by Eve.
For this problem, it is assumed that Alice and Bob share a common random number $a \in \cA$,
and Eve has another random number $e \in \cE$, which is correlated to the random number $a$. 
The task is to extract a common random number 
$f(a)$ from the random number $a \in \cA$, which is almost independent of 
Eve's random number $e \in \cE$.
Here, Alice and Bob are only allowed to apply the same function $f$ to the common random number $a \in \cA$.
In order to discuss this problem, 
for $s \in [0,1]$, we define the functions
\begin{align*}
\tilde{H}_{1+s}(X|P^X)&:=-\log \sum_{x} P^{X}(x)^{1+s}  \\
\tilde{H}_{1+s}(X|Y|P^{X,Y})
&:=-\log \sum_{x,y} P^Y(y)P^{X|Y}(x|y)^{1+s}  \\
&=-\log \sum_{x,y} P^{X,Y}(x,y)^{1+s} P^Y(y)^{-s}. 
\end{align*}
Using these functions, we can define 
R\'{e}nyi entropy of order $1+s$
\begin{align*}
H_{1+s}(X|P^X):=\frac{\tilde{H}_{1+s}(X|P^X)}{s}
\end{align*}
and the conditional R\'{e}nyi entropy of order $1+s$:
\begin{align*}
H_{1+s}(X|Y|P^{X,Y}):=\frac{\tilde{H}_{1+s}(X|Y|P^{X,Y})}{s}.
\end{align*}
If there is no possibility for confusion,
$P^{X,Y}$ is omitted.

Now, we focus on an ensemble of the functions $f_{\bX}$ from 
$\cA$ to $\{1, \ldots, M\}$, where $\bX$ denotes a random variable describing 
the stochastic behavior of the function $f$.
An ensemble of the functions $f_{\bX}$ is called universal$_2$ 
when it satisfies the following condition\cite{Carter}:
\begin{condition}\Label{C1}
$\forall a_1 \neq \forall a_2\in \cA$,
the probability that $f_{\bX}(a_1)=f_{\bX}(a_2)$ is
at most $\frac{1}{M}$.
\end{condition}
We sometimes require 
the following additional condition:
\begin{condition}\Label{C12}
For any $\bX$,
the cardinality of $f_{\bX}^{-1}\{i\}$
does not depend on $i$.
\end{condition}
This condition will be used in Section III.

Indeed, when the cardinality $|\cA|$ is a power of a prime power $q$
and $M$ is another power of the same prime power $q$,
an ensemble $\{f_{\bX}\}$ 
satisfying the both conditions
is given by the 
the concatenation of Toeplitz matrix and the identity 
$(\bX,I)$\cite{Krawczyk}
only with $\log_q |\cA|-1$ random variables taking values in the finite filed $\bF_q$.
That is, the matrix $(\bX,I)$ has small complexity.
The construction and its proof are given in Appendix \ref{stoep}.

When $M$ is an arbitrary integer and 
the cardinality $|\cA|$ is an arbitrary multiple of $M$,
an ensemble $\{f_{\bX}\}$ 
satisfying the both conditions
is given in the following way.
First, we fix a function $f$ from $\cA$ to $\{1, \ldots, M\}$
such that the cardinality $|f^{-1}\{i\}|$ is $\frac{|\cA|}{M}$.
We randomly choose an permutation $\sigma \in S_{\cA}$ on $\cA$ with the uniform distribution,
where $S_{\cA}$ denotes the set of permutation on $\cA$.
So, we can make a random function $\{f\circ \sigma \}_{\cA}$.
This ensemble satisfies the both conditions.

As is shown in the Appendix \ref{s7}, we obtain the following theorem.
\begin{thm}\Label{thm1}
When the ensemble of the functions $\{f_{\bX}\}$ is 
universal$_2$, it satisfies 
\begin{align}
\rE_\bX H(f_{\bX}(A)|E|P^{A,E}) 
&\ge 
\log M - 
\frac{M^s e^{-\tilde{H}_{1+s}(A|E|P^{A,E})}}{s}\nonumber\\
&=
\log M - 
\frac{e^{s(\log M-H_{1+s}(A|E|P^{A,E}))}}{s}
\Label{5-14-1}
\end{align}
for $0 < \forall s \le 1$.
\end{thm}
Note that Bennett et al \cite{BBCM} proved this inequality for the case of $s=1$.

Since
the mutual information
\begin{align*}
I(f_{\bX}(A):E|P^{A,E}) := H(f_{\bX}(A)|P^{A}) - H(f_{\bX}(A)|E|P^{A,E})  
\end{align*}
is bounded by $\log M -H(f_{\bX}(A)|E|P^{A,E})$,
we obtain
\begin{align}
\rE_{\bX} I(f_{\bX}(A):E|P^{A,E})
\le
\frac{M^s e^{-\tilde{H}_{1+s}(A|E|P^{A,E})}}{s}, 0 < s \le 1.
\Label{4-27-4}
\end{align}

This inequality implies the following theorem.
\begin{thm}
There exists a function $f$ from $\cA$ to $\{1, \ldots, M\}$ 
such that
\begin{align}
I(f(A):E)
&\le  \frac{M^s e^{-\tilde{H}_{1+s}(A|E|P^{A,E})}}{s}\nonumber \\
&=\frac{e^{s(\log M-H_{1+s}(A|E|P^{A,E}))}}{s},
\quad 0 \le \forall s \le 1.
\Label{4-25-1}
\end{align}
\end{thm}
In the following, we mainly use the quantity $\tilde{H}_{1+s}(A|E|P^{A,E})$
instead of $H_{1+s}(A|E|P^{A,E})$.
because the usage of $H_{1+s}(A|E|P^{A,E})$ requires 
more complicated calculation.

Next, we consider the case when our distribution $P^{A_n E_n}$ 
is given by the $n$-fold independent and identical distribution of 
$P^{AE}$, i.e, $(P^{A,E})^n$.
Ahlswede and Csisz\'{a}r \cite{AC93} showed that
the optimal generation rate
\begin{align*}
&G(P^{AE}) \\
:=&
\sup_{\{(f_n,M_n)\}}
\left\{
\lim_{n\to\infty} \frac{\log M_n}{n}
\left|
\begin{array}{l}
\displaystyle
\lim_{n\to\infty} \frac{I(f_n(A_n):E_n)}{n}=0 \\
\displaystyle
\lim_{n\to\infty} \frac{H(f_n(A_n))}{\log M_n}=1
\end{array}
\right. \right\}
\end{align*}
equals the conditional entropy $H(A|E)$.
That is, the generation rate $R= \lim_{n\to\infty} \frac{\log M_n}{n}$
is smaller than $H(A|E)$,
Eve's information $I(f_n(A_n):E_n)$ goes to zero.
In order to treat the speed of this convergence,
we focus on the supremum of  
the {\it exponentially decreasing rate (exponent)} of 
$I(f_n(A_n):E_n)$ for a given $R$
\begin{align*}
&e_I(P^{AE}|R) \\
:=& \!\!\!
\sup_{\{(f_n,M_n)\}}\!\!
\left\{\!
\lim_{n\to\infty} \! \frac{-\log I(f_n(A_n):E_n)}{n}
\!\left|\!
\lim_{n\to\infty} \!\!\frac{-\log M_n}{n} 
\!\le\! R\!
\right. \right\}.
\end{align*}
Since the relation
$\tilde{H}_{1+s}(A_n|E_n|(P^{A,E})^n)= n \tilde{H}_{1+s}(A|E|P^{A,E})$ holds,
the inequality (\ref{4-25-1}) implies that
\begin{align}
e_I(P^{AE}|R) 
&\ge \max_{0 \le s \le 1}  \tilde{H}_{1+s}(A|E|P^{A,E})-sR\nonumber \\
&= \max_{0 \le s \le 1}  s (H_{1+s}(A|E|P^{A,E}-R)
\Label{4-16-4}
\end{align}
Since $\left.\frac{d}{ds}\tilde{H}_{1+s}(A|E|P^{A,E})\right|_{s=0}=H(A|E)$,
Eve's information $I(f_n(A_n):E_n)$ exponentially goes to zero for $R<H(A|E)$.

\subsection{Method based on smooth min-entropy}\Label{s42}
R\'{e}nyi entropy of order $2$ $H_2(A|E|P^{A,E})$ is bounded by 
the min-entropy
\begin{align*}
H_{\min}(A|E|P^{A,E}):= 
\min_{a,e:P^{A,E}(a,e)>0} -\log P^{A|E}(a|e),
\end{align*}
i.e., the inequality
\begin{align*}
H_2(A|E|P^{A,E})\ge H_{\min}(A|E|P^{A,E})
\end{align*}
holds. Then, 
(\ref{5-14-1}) with $s=1$ yields that
\begin{align}
& \rE_{\bX} 
\log M + H(E |P^{A,E}) -H(f_{\bX}(A) E|P^{A,E})\nonumber  \\
= &
\rE_{\bX} 
\log M -H(f_{\bX}(A)|E|P^{A,E}) \nonumber \\
\le &
M e^{-H_{\min}(A|E|P^{A,E})}
\Label{4-14-1}.
\end{align}
Renner and Wolf \cite{RW} introduced the smooth min-entropy:
\begin{align}
& H_{\min}^{\epsilon}(A|E|P^{A,E}) \nonumber \\
:=&
\max_{\Omega:P^{A,E}(\Omega) \ge 1-\epsilon }
\min_{(a,e)\in \Omega} -\log P^{A|E}(a|e).
 \Label{4-14-2}
\end{align}
for $\epsilon \ge 0$.
This definition is different from that of Renner \cite{Ren05}.
Modifying the discussion by Renner and Wolf \cite{RW},
we can derive another upper bound of $\rE_{\bX} I(f_{\bX}(A):E)$
based on the smooth min-entropy 
$H_{\min}^{\epsilon}(A|E|P^{A,E})$ in the following way.

Using the variational distance $d(P^X,\tilde{P}^X)$:
\begin{align*}
d(P^X,\tilde{P}^X) :=
\sum_{x}|P^X(x)-\tilde{P}^X(x)|,
\end{align*}
we have the continuity of the Shannon entropy 
in the following sense:
When $d (P^X,\tilde{P}^X) \le \frac{1}{e}$,
the function
\begin{align*}
\eta(x,a):=- x\log x+x a
\end{align*}
satisfies the following inequality:
\begin{align*}
& |H(X|\tilde{P}^X)- H(X|P^X)| \\
\le & \eta( d (P^X,\tilde{P}^X),\log |{\cal X}|).
\end{align*}
Based on the variational distance,
we define the following modification:
\begin{align}
& \hat{H}_{\min}^{\epsilon}(A|E|P^{A,E}) \nonumber \\
:=&
\max_{\tilde{P}^{A,E}}\{H_{\min}(A|E|\tilde{P}^{A,E})|
d(\tilde{P}^{A,E},P^{A,E}) \le \epsilon
\},\Label{4-16-1}
\end{align}
where $\tilde{P}^{A,E}$ is a probability distribution.

For $0<\epsilon <1/2$,
we choose $\Omega$ satisfying the condition in (\ref{4-14-2}).
Then, 
$p_{\max}^{A|E}(\Omega):= \max_{(a,e)\in \Omega}  P^{A|E}(a|e) \ge \frac{1}{|A|}$.
We define the joint distribution $\tilde{P}^{A,E}(a,e)$ 
satisfying $\tilde{P}^{E}(e)={P}^{E}(e)$ in the following way.
For this purpose, it is sufficient to define the conditional distribution $\tilde{P}^{A|E}(a|e)$
for all $e$.
When $(a,e)\in \Omega$, 
the conditional distribution $\tilde{P}^{A|E}(a|e)$ is defined by
\begin{align*}
\tilde{P}^{A|E}(a|e):=
\left\{
\begin{array}{ll}
{P}^{A|E}(a|e) &  
\hbox{ if } {P}^{A|E}(a|e) \le p_{\max}^{A|E}(\Omega)\\
p_{\max}^{A|E}(\Omega) & 
\hbox{ if } 
{P}^{A|E}(a|e) > p_{\max}^{A|E}(\Omega) .
\end{array}
\right.
\end{align*}
When $(a,e)\notin \Omega$, 
we define $\tilde{P}^{A|E}(a|e)$ satisfying that
\begin{align*}
& {P}^{A|E}(a|e) 
\le 
\tilde{P}^{A|E}(a|e)
\le \frac{1}{|A|} , \\
& \sum_{(a,e)\notin \Omega}
(\tilde{P}^{A|E}(a|e)-{P}^{A|E}(a|e)) \\
= &
\sum_{(a,e)\in \Omega}
({P}^{A|E}(a|e)-\tilde{P}^{A|E}(a|e)).
\end{align*}
Then, $d(\tilde{P}^{A,E},P^{A,E}) \le 2 \epsilon$.
Since 
\begin{align*}
{H}_{\min}(A|E|\tilde{P}^{A,E}) 
\ge - \log 
p_{\max}^{A|E},
\end{align*}
we have
\begin{align*}
\hat{H}_{\min}^{2 \epsilon}(A|E|P^{A,E})
\ge {H}_{\min}^{\epsilon}(A|E|P^{A,E}).
\end{align*}

When $\tilde{P}^{A,E}$ satisfies the condition given in (\ref{4-16-1}),
\begin{align*}
& |
(H(E |P^{A,E}) -H(f_{\bX}(A) E|P^{A,E})) \\
&-
(H(E |\tilde{P}^{A,E}) -H(f_{\bX}(A) E|\tilde{P}^{A,E}))
| \\
\le &
2 \eta( \epsilon, \log  |{\cal A}| \cdot M).
\end{align*}
Hence, 
\begin{align*}
& \rE_{\bX} 
I(f_{\bX}(A):E|P^{A,E}) \\
\le &
 \rE_{\bX} 
\log M + H(E |P^{A,E}) -H(f_{\bX}(A) E|P^{A,E})\\
\le &
 \rE_{\bX} 
\log M + H(E |\tilde{P}^{A,E}) -H(f_{\bX}(A) E|\tilde{P}^{A,E})\\
& + 2 \eta( \epsilon, \log  |{\cal A}| \cdot M) \\
\le &
M e^{-H_{\min}(A|E|\tilde{P}^{A,E})} 
+ 2 \eta( \epsilon,\log  |{\cal A}| \cdot M )\\
\le &
M e^{-\hat{H}_{\min}^{\epsilon}(A|E|{P}^{A,E})} 
+ 2 \eta( \epsilon,\log  |{\cal A}| \cdot M )\\
\le &
M e^{-{H}_{\min}^{\epsilon/2}(A|E|{P}^{A,E})} 
+ 2 \eta( \epsilon,\log  |{\cal A}| \cdot M ).
\end{align*}
Thus, we obtain an alternative bound of
$\rE_{\bX} I(f_{\bX}(A):E|P^{A,E})$ as follows.
\begin{align}
& \rE_{\bX} 
I(f_{\bX}(A):E|P^{A,E}) \nonumber \\
\le &
\overline{I}_{\min,M}(A|E|{P}^{A,E}) \nonumber \\
:= &
\min_{1/4> \epsilon >0}
M e^{-{H}_{\min}^{\epsilon}(A|E|{P}^{A,E})} 
+ 2 \eta( 2\epsilon,\log  |{\cal A}| \cdot M )\nonumber \\
\le & \min_{R' \ge \log 4|{\cal A}|}
M e^{-R'} \nonumber \\
& + 2 \eta(2 P^{A,E}\{ P^{A|E}(a|e) \ge e^{-R'} \},\log  |{\cal A}| \cdot M ).
\Label{4-16-2}
\end{align}
Using (\ref{4-16-2}),
we can evaluate $e_I(P^{AE}|R)$ as follows.
\begin{align*}
e_I(P^{AE}|R)
\ge 
\lim_{n\to \infty}
\frac{-1}{n}\log
\overline{I}_{\min,e^{nR}}(A|E|({P}^{A,E})^n) 
\end{align*}
Cram\'{e}r Theorem yields that
\begin{align*}
& \lim_{n\to \infty}
\frac{-1}{n}\log
(P^{A,E})^n \{ (P^{A|E})^n (a|e) \ge e^{-n R'} \} \\
= &
\max_{s \ge 0}\tilde{H}_{1+s} (A|E|{P}^{A,E}) -sR'.
\end{align*}
Thus,
\begin{align*}
 \lim_{n\to \infty}
\frac{-1}{n}\log P_n(R') 
=
\max_{s \ge 0} \tilde{H}_{1+s} (A|E|{P}^{A,E}) -sR'.
\end{align*}
where
\begin{align*}
&P_n(R')\\
:=&
\eta(
2 (P^{A,E})^n \{ (P^{A|E})^n (a|e) \ge e^{-n R'} \}
,\log  |{\cal A}|^n e^{nR} ) .
\end{align*}
Therefore,
\begin{align*}
& \lim_{n\to \infty}
\frac{-1}{n}\log
\overline{I}_{\min,e^{nR}}(A|E|({P}^{A,E})^n)  \\
=&
\max_{R': R'\ge R}
\min\{
\max_{s \ge 0}
\tilde{H}_{1+s} (A|E|{P}^{A,E}) -sR'
,R' -R\}.
\end{align*}
$\max_{s \ge 0} \tilde{H}_{1+s} (A|E|{P}^{A,E}) -sR'$ 
is continuous and monotone decreasing concerning $R'$
and 
$R' -R$ 
is continuous and monotone increasing concerning $R'$.
Thus,
the above maximum is attained when 
$\max_{s \ge 0} \tilde{H}_{1+s} (A|E|{P}^{A,E}) -sR'=R' -R$.
Let $s_0$ be the parameter $s$ attaining the above.
Then, 
$\tilde{H}_{1+s_0} (A|E|{P}^{A,E}) -s_0R'= R'-R$
and $\frac{d}{ds }\tilde{H}_{1+s_0} (A|E|{P}^{A,E})|_{s=s_0}= R'$.
Thus,  
\begin{align}
& \max_{R': R'\ge R}
\min\{
\max_{s \ge 0} \tilde{H}_{1+s} (A|E|{P}^{A,E}) -sR',
R' -R\} \nonumber \\
= &
\frac{1}{1+s_0} \tilde{H}_{1+s_0} (A|E|{P}^{A,E}) 
-\frac{s_0}{1+s_0}R \nonumber \\
= &
\max_{s \ge 0} \frac{1}{1+s}  \tilde{H}_{1+s} (A|E|{P}^{A,E}) 
- \frac{s}{1+s}R \Label{n4-16-6} \\
= &
\max_{s \ge 0} \frac{s}{1+s}  
(H_{1+s} (A|E|{P}^{A,E}) - R),
\Label{4-16-6}
\end{align}
where the equation (\ref{n4-16-6}) can be checked by taking the derivative.
This value is smaller than the bound given by (\ref{4-16-4}).
One might want to apply the formula
\begin{align*}
H_{\min}^\epsilon(A) \ge
H_{1+s}(A) +\frac{\log \epsilon}{s}
\end{align*}
given by Renner and Wolf\cite{RS2}
to the evaluation of $ \overline{I}_{\min,M}(A|E|{P}^{A,E})$.
However, this application does not simplify our derivation.
So, we do not apply this formula.

\section{The wire-tap channel in a general framework}\Label{s2}
Next, we consider the wire-tap channel model, in which
the eavesdropper (wire-tapper), Eve 
and the authorized receiver Bob
receive information from the authorized sender Alice.
In this case, in order for
Eve to have less information,
Alice chooses a suitable encoding.
This problem is formulated as follows.
Let $\cY$ and $\cZ$ be the probability spaces of 
Bob and Eve,
and $\cX$ be the set of alphabets sent by Alice.
Then, the main channel from Alice to Bob
is described by $W^B:x \mapsto W^B_x$,
and the wire-tapper channel from Alice to 
Eve is described by $W^E:x \mapsto W^E_x$.
In this setting,
Alice chooses $M$ 
distributions $Q_1, \ldots, Q_M$ on $\cX$,
and she generates $x\in \cX$ subject to $Q_i$
when she wants to send the message $i \in \{1, \ldots, M\}$.
Bob prepares $M$ disjoint subsets
$\cD_1,\ldots, \cD_M$ of $\cY$ and 
judges that a message is $i$ if $y$ belongs to $\cD_i$.
Therefore, the triplet $(M,\{Q_1, \ldots, Q_M\},
\{\cD_1,\ldots, \cD_M\})$ is called a
code, and is described by $\Phi$.
Its performance is given by the following three quantities.
The first is the size $M$, which is denoted by $|\Phi|$.
The second is the average 
error probability $\epsilon_B(\Phi)$:
\begin{align*}
\epsilon_B(\Phi)\defeq
\frac{1}{M} \sum_{i=1}^M  W_{Q_i}^B (\cD_i^c),
\end{align*}
and the third is Eve's information
regarding the transmitted message $I_E(\Phi)$:
\begin{align*} 
I_E(\Phi) \defeq \sum_i \frac{1}{M} D(  W_{Q_i}^E\| W^E_{\Phi}),\quad
W^E_{\Phi}  \defeq \sum_i \frac{1}{M} W_{Q_i}^E.
\end{align*}
In order to calculate these values, 
we introduce the following quantities.
\begin{align*}
\phi(s|W,p) &:= \log 
\sum_y \left( \sum_x p(x)
(W_x(y))^{1/(1-s)}\right)^{1-s} \\
\psi(s|W,p) &:= \log 
\sum_y \left( \sum_x p(x)
(W_x(y))^{1+s}\right)W_p(y)^{-s},
\end{align*}
where
$W_p(y):= \sum_x p(x)W_x(y)$.
The following lemma gives the properties of these quantities.

\begin{lem}\red{[13]}\Label{l1}
The function $p \mapsto 
e^{\phi(s|W,p)}$ 
is convex for $s\in [-1,0]$, and is concave for $s\in [0,1]$.
\end{lem}
\begin{proof}
\red{The convexity and concavity of $p \mapsto 
e^{\phi(s|W,p)}$ follow} from the convexity and concavity of
$x^{1-s}$ for the respective parameter $s$.
\end{proof}

Now, using the functions $\phi(s)$ and $\psi(s)$,
we make a code for the wire-tap channel based on the random coding method.
For this purpose, we make a protocol to share a random number.
First, we generate the random code $\Phi(\bY)$ with size $LM$,
which is described by the
$LM$ independent and identical random
variables $\bY$ subject to the distribution $p$ on $\cX$.
For integers \red{$k=1, \ldots, LM$} 
let \red{$\cD_{k}'(\bY)$} be the maximum likelihood decoder
of the code $\Phi(\bY)$.
Gallager \cite{Gal} 
showed that the ensemble expectation of the average error 
probability 
concerning decoding the input message $A$ is less than
$(ML)^{s}e^{\phi(-s|W^B,p)}$ for $0 \le s \le 1$.
Here, we choose a function 
$f_{\bX}$ from a function ensemble 
$\{f_{\bX}\}$ 
satisfying Conditions \ref{C1} and \ref{C12}.
After sending the random variable $A$ taking values in the set with the cardinality $ML$, 
Alice and Bob apply the function $f_{\bX}$ 
to the random variable $A$ and generate another piece of data of size $M$.
Then, Alice and Bob 
share random variable $f_{\bX}(A)$ with size $M$.
This protocol is denoted by $\Phi(\bX,\bY)'$ 

Let $E$ be the random variable of the output of Eve's channel $W^E$,
and
$f_{\Phi(\bY)}$
be the map defined by the code $\Phi(\bY)$
from the message space $\{1, \ldots, ML\}$ to $\cX$.
\red{Then as is shown in Appendix \ref{a4}, 
we obtain 
\begin{align}
 \rE_{\bY} E_{\bX|\bY} I_E(\Phi(\bX,\bY)') 
\le 
\frac{1}{s L^s} e^{\psi(s|W,p)}
\quad 0 < s \le 1.
\Label{2-25-1}
\end{align}
}
Now, we make a code for wire-tap channel by modifying the above protocol $\Phi(\bX,\bY)'$.
First, we choose the distribution $Q_i$ to be the uniform distribution on $f_{\bX}^{-1}\{i\}$.
When Alice wants to send the message $i$,
before sending the random variable $A$, 
Alice generates the random number $A$ subject to the distribution $Q_i$.
Alice sends the random variable $A$.
Bob recovers the random variable $A$ and Applies the function $f_{\bX}$.
Then, Bob decodes Alice's message $i$, and this code for wire-tap channel $W^B,W^E$
is denoted by $\Phi(\bX,\bY)$.
Since Condition \ref{C12} guarantees that
the cardinality $|f_{\bX}^{-1}\{i\}|$ does not depend on $i$,
the protocol $\Phi(\bX,\bY)$ has the same performance as the above protocol $\Phi(\bX,\bY)'$.

Finally, 
we consider what code is derived from the above random coding discussion.
Using the Markov inequality, we obtain
\begin{align*}
\rP_{\bX,\bY} 
\{ \epsilon_B(\Phi(\bX,\bY)) \le 2 \rE_{\bX,\bY} \epsilon_B(\Phi(\bX,\bY)) \}^c
&\,< \frac{1}{2} \\
\rP_{\bX,\bY} 
\{ I_E(\Phi(\bX,\bY)) \le 2 \rE_{\bX,\bY} I_E(\Phi(\bX,\bY))\}^c
&\,< \frac{1}{2} .
\end{align*}
Therefore, the existence of a good code is guaranteed in the following way.
That is, we give the concrete performance of a code 
whose existence is shown in the above random coding method.

\begin{thm}\Label{3-6}
There exists a code $\Phi$ for any
integers $L,M$,
and any probability distribution $p$ on $\cX$
such that
\begin{align}
|\Phi| &=M \nonumber \\
\epsilon_B(\Phi) & \le 2 \min_{0\le s\le 1}(ML)^{s}e^{\phi(-s|W^B,p)}
\Label{3-8-1}\\
I_E(\Phi) & \le  2
\min_{0 \le s \le 1}
\frac{e^{\psi(s|W^E,p)}}{L^s s}.
\Label{7-1-2} 
\end{align}
\end{thm}
In fact, Hayashi \cite{Hayashi} proved a similar result when 
the right hand side of (\ref{7-1-2}) is replaced by
$2 \min_{0 \le s \le 1/2}\frac{e^{\phi(s|W^E,p)}}{L^s s}$.


\vspace{2ex}
In the $n$-fold discrete memoryless channels $W^{B_n}$ and $W^{E_n}$ 
of the channels $W^B$ and $W^E$,
the additive equation
$\phi(s|W^{B_n},p)= n \phi(s|W^B,p)$ 
holds.
Thus, there exists a code $\Phi_n$ for any
integers $L_n,M_n$,
and any probability distribution $p$ on $\cX$
such that
\begin{align}
|\Phi_n| &=M_n \nonumber \\
\epsilon_B(\Phi) & \le 2
\min_{0\le s\le 1}
(M_n L_n)^{s}e^{n \phi(-s|W^B,p)}
 \nonumber
\\
I_E(\Phi_n) & \le  2
\min_{0 \le s \le 1}
\frac{e^{n \psi(s|W^E,p)}}{L_n^s s}.
\Label{7-1-2-a} 
\end{align}
Since 
$\lim_{s \to 0} \frac{\psi(s|W^{E},p)}{s}= I(p:W^E)$,
the rate $\max_p I(p:W^B)-I(p:W^E)$ can be asymptotically 
attained.

When the sacrifice information rate is $R$, i.e., $L_n\cong e^{nR}$,
the decreasing rate of Eve's information is 
greater than
$e_{\psi}(R|W^E,p):=\max_{0\le s\le1} s R-\psi(s|W^E,p)$.
Hayashi \cite{Hayashi} derived 
another lower bound of this exponential decreasing rate
$e_{\phi}(R|W^E,p):=\max_{0 \le s \le 1/2} s R-\phi(s|W^E,p)$.

\section{Comparison with existing bound}\Label{s6}
Now, we compare the two upper bounds $\frac{e^{\psi(s|W^E,p)}}{L^s s}$ and 
$\frac{e^{\phi(s|W^E,p)}}{L^s s}$ \red{for $0 < s \le 1$}.
H\"{o}lder inequality
with the measurable space $({\cal X},p)$ is given as
\begin{align*}
& |\sum_{x\in {\cal X}} p(x) X(x)Y(x)| \\
\le &
(\sum_{x\in {\cal X}} p(x) |X(x)|^{\frac{1}{1-s}})^{1-s}
(\sum_{x\in {\cal X}} p(x) |Y(x)|^{\frac{1}{s}})^{s}.
\end{align*}
Using this inequality, we obtain
\begin{align*}
& 
\sum_x p(x)
(W_x(y))^{1+s} 
W_p(y)^{-s}\\
= &
\sum_x p(x)
W_x(y)
(\frac{W_x(y)}{W_p(y)})^{s}
\\
\le &
\left( \sum_x p(x)
(W_x(y))^{\frac{1}{1-s}}\right)^{1-s}
\left( \sum_x p(x)
\frac{W_x(y)}{W_p(y)}
\right)^{s} \\
= &
\left( \sum_x p(x)
(W_x(y))^{\frac{1}{1-s}}
\right)^{1-s}.
\end{align*}
Taking the summand concerning $y$, 
we obtain
\begin{align}
e^{\psi(s|W^E,p)} \le e^{\phi(s|W^E,p)}.
\Label{2-26-1}
\end{align}
That is, our upper bound is better than that given by \cite{Hayashi}.
Thus, $e_{\psi}(R|W^E,p)\ge e_{\phi}(R|W^E,p)$.

Next, 
in order to consider the case 
when the privacy amplification rate $R$ is close to the mutual information
$I(p:W)$, 
we treat the difference between these bounds
with the limit $s \to 0$.
In this case, we take their Taylor expansions as
follows.
\begin{align*}
& \sum_{x,y}
p_x W_x(y)^{1+s}W_p(y)^{-s} \\
\cong &
1+ I(p:W)s +I_2(p:W)s^2 +I_3(p:W)s^3 \\
&
\sum_{y} \left(\sum_{x}p_x W_x(y)^{\frac{1}{1-s}}\right)^{1-s} \\
\cong &
1+ I(p:W)s +I_2(p:W)s^2 +
(I_3(p:W)+\tilde{I}_3(p:W))s^3 ,
\end{align*}
where
\begin{align*}
I_2(p:W)&:=
\frac{1}{2}
\sum_{x,y}p_x W_x(y)
(\log W_x(y)-\log W_p(y))^2\\
I_3(p:W)&:=
\frac{1}{6}
\sum_{x,y}p_x W_x(y)
(\log W_x(y)-\log W_p(y))^3\\
\tilde{I}_3(p:W)&:=
\frac{1}{2}\sum_y
\Bigl(
\sum_x p_x W_x(y)(\log W_x(y))^2  \\
&\hspace{10ex} - \frac{(\sum_x p_x W_x(y)\log W_x(y))^2}{W_p(y)}
\Bigr).
\end{align*}
Indeed, applying the Schwarz inequality 
to the inner product $\langle f,g \rangle:=
\sum_x p_x W_x(y) f(y) g(y)$, we obtain
\begin{align*}
&(\sum_x p_x W_x(y)(\log W_x(y))^2 )
\cdot 
(\sum_x p_x W_x(y) ) \\
\ge &
(\sum_x p_x W_x(y)\log W_x(y))^2.
\end{align*}
Since
$\sum_x p_x W_x(y)=W_p(x)$,
this inequality implies that 
$\tilde{I}_3(p:W) \ge 0$.
That is,
$e^{\psi(s|W^E,p)}$ is smaller than 
$e^{\phi(s|W^E,p)}$ only in the third order when $s$ is small.

Next, we consider a more specific case.
A channel $W^E$ is called {\it additive}
when there exists a distribution such that 
$W^E_x(z)=P(z-x)$.
In this case,
$\frac{e^{\psi(s|W^E,p)}}{L^s s}$ can be simplified as follows.
When ${\cal X}={\cal Z}$ and ${\cal X}$ is a module and $W_{x}(z)=W_{0}(z-x)
=P(z-x)
$, the channel $W$ is called additive.
\red{The quantities $e_{\psi}(R|W^E,p_{\mix})$ and $e_{\phi}(R|W^E,p_{\mix})$ are characterized as follows.}
Since
\begin{align}
e^{\psi(s|W^E,p_{\mix})}
&=
|{\cal X}|^{s} e^{-\tilde{H}_{1+s}(X|P)}\Label{12-26-2}\\
e^{\phi(t|W^E,p_{\mix})}
&=
|{\cal X}|^{t} e^{- (1-t)\tilde{H}_{1+\frac{t}{1-t}}(X|P)},
\Label{12-26-2-b}
\end{align}
we obtain 
\begin{align*}
& e_{\psi}(R|W^E,p_{\mix})
=\max_{0 \le s \le 1}
s (R-\log |{\cal X}|)
+ \tilde{H}_{1+s}(X|P) \\
=& \max_{0 \le s \le 1}
s (R-\log |{\cal X}|+ H_{1+s}(X|P) ) \\
\ge &
\max_{0 \le s \le 1}
\frac{s (R-\log |{\cal X}|)+ \tilde{H}_{1+s}(X|P) }{1+s} 
\\
=& \max_{0 \le s \le 1}
\frac{s(R-\log |{\cal X}|+ H_{1+s}(X|P) )}{1+s} 
=
e_{\phi}(R|W^E,p_{\mix}),
\end{align*}
where $t=\frac{s}{1+s}$.
Fig. \ref{f1} shows the comparison of 
$e_{\psi}(R|W^E,p_{\mix})$ and $e_{\phi}(R|W^E,p_{\mix})$
with $e_{\psi,2}(R|W^E,p_{\mix}):=  (R-\log |{\cal X}|)+ H_{2}(X|P) $,
which is directly obtained from Bennett et al\cite{BBCM}.
When $R-\log |{\cal X}| \ge -\frac{d}{ds} \tilde{H}_{1+s}(X|P)|_{s=1}$,
$e_{\psi}(R|W^E,p_{\mix})= e_{\psi,2}(R|W^E,p_{\mix})$.

\begin{figure}[htbp]
\begin{center}
\scalebox{1.0}{\includegraphics[scale=0.3]{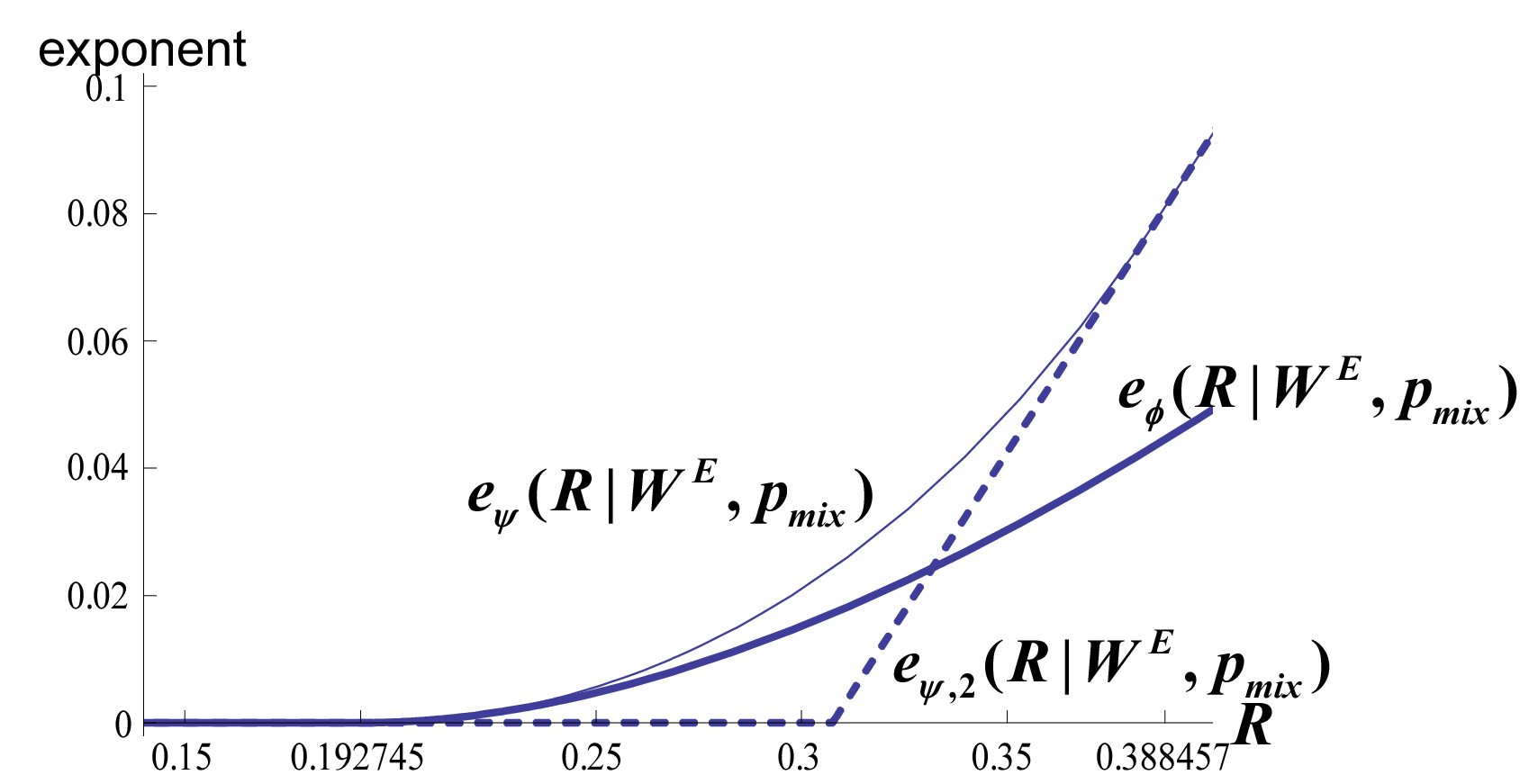}}
\end{center}
\caption{Normal line: $e_{\psi}(R|W^E,p_{\mix})$ (The present paper),
Thick line: $e_{\phi}(R|W^E,p_{\mix})$ (Hayashi\cite{Hayashi}),
Dashed line: $e_{\psi,2}(R|W^E,p_{\mix})$ (Bennett et al\cite{BBCM}).
$p=0.2$, $\log 2 - h(p)=0.192745$,
$\log |{\cal X}| -\frac{d}{ds} \tilde{H}_{1+s}(X|P)|_{s=1}=0.388457$.}
\Label{f1}
\end{figure}%

Next, we consider a more general case.
Eve is assumed to have two random variables $z \in {\cal X}$ and $z'$.
The first random variable $z$ is the output of an additive channel 
depending on the second variable $z'$.
That is, the channel $W_x^E(z,z')$ can be written as
$W_x^E(z,z')=P^{X,Z'}(z-x,z')$, where $P^{X,Z'}$ is a joint distribution.
Hereinafter, this channel model is called a general additive channel.
This channel is also called 
a regular channel\cite{DP}.
For this channel model,
the inequality $e_{\psi}(R|W^E,p_{\mix})\ge e_{\phi}(R|W^E,p_{\mix})$
holds because
\begin{align}
e^{\psi(s|W^E,p_{\mix})}
&=
|{\cal X}|^{s} e^{-\tilde{H}_{1+s}(X|Z'|P^{X,Z'})}\Label{12-26-3}\\
e^{\phi(t|W^E,p_{\mix})}
&=
|{\cal X}|^{t} e^{- (1-t)\tilde{H}_{1+\frac{t}{1-t}}(X|Z'|P^{X,Z'})}.\nonumber
\end{align}

\section{Wire-tap channel with linear coding}\Label{s3}
In a practical sense, 
we need to take into account the decoding time.
For this purpose, we often restrict our codes to linear codes.
In the following, we consider the case where 
the sender's space $\cX$ has the structure of a module.
First, we regard a submodule $C_1\subset \cX$ as an encoding for
the usual sent message,
and focus on its decoding $\{\cD_x\}_{x\in C_1}$ by the authorized receiver.
We construct a code for a wire-tap channel
$\Phi_{C_1,C_2}= (|C_1/C_2|,
\{Q_{[x]}\}_{[x]\in C_1/C_2},
\{\cD_{[x]}\}_{[x]\in C_1/C_2})$
based on a submodule $C_2$ of $C_1$ as follows.
The encoding $Q_{[x]}$ is given as the uniform distribution 
on the coset $[x]:=x+C_2$, and 
the decoding $\cD_{[x]}$ is given as the subset 
$\cup_{x'\in x+C_2} \cD_{x'}$.
Next, we assume that a submodule 
$C_2(\bX)$ of $C_1$ with cardinality $|C_2(\bX)|=L$
is generated by a random variable $\bX$ satisfying the following condition.

\begin{condition}\Label{C2}
Any element $x \neq 0 \in C_1$
is included in $C_2(\bX)$ with probability at most 
$\frac{L}{|C_1|}$.
\end{condition}

Then, the performance of the constructed code is evaluated by the following theorem.
\begin{thm}
Choose the subcode $C_2(\bX)$ according to Condition \ref{C2}.
We construct the code $\Phi_{C_1,C_2(\bX)}$ by choosing the distribution 
$Q_{[x]}$ to be the uniform distribution on $[x]$ for $[x]\in C_1/C_2(\bX)$.
Then, we obtain
\begin{align}
\rE_{\bX} I_E(\Phi_{C_1,C_2(\bX)}) 
\le &
\frac{e^{\psi(s|W^E,P_{\mix, C_1})}}{L^s s}
\quad 0 < \forall s <1,\Label{4-27-1}
\end{align}
where $P_{\mix,S}$ is the uniform distribution on the subset $S$.
\end{thm}

\begin{proof}
This inequality can be shown by (\ref{4-27-4}) as follows.
Now, we define 
the joint distribution 
$P(x,z):= P_{\mix, C_1}(x)W^E_{x}(z)$.
The choice of $Q_{[x]}$ corresponds to 
a hashing operation satisfying Condition \ref{C1}.
Then, 
(\ref{4-27-4}) yields that
$\rE_{\bX} I_E(\Phi_{C_1,C_2(\bX)}) $
is bounded by 
$\frac{|C_1|^s
\sum_{x,z} P(z,x)^{1+s}P(z)^{-s}
}{\red{L^s} s}
= \frac{e^{\psi(s|W^E,P_{\mix, C_1})}}{L^s s}$,
which implies (\ref{4-27-1}).
\end{proof}

\red{Next, we assume that a submodule 
$C_1(\bY)$ of with cardinality $|C_1(\bY)|=ML$
is generated by a random variable $\bY$ satisfying the following condition.}

\red{\begin{condition}\Label{C4}
The relation $|C_1(\bY)|=ML$ always holds.
Any element $x \neq 0 \in {\cal X}$
is included in $C_1(\bY)$ with probability at most 
$\frac{M}{|{\cal X}|}$.
\end{condition}}

\red{Choose the subcode $C_1(\bY)$ and $C_2(\bX)$ according to Conditions \ref{C4} and \ref{C2}.
Then, as is shown in Appendix \ref{a5}, we obtain
\begin{align}
\rE_{\bX,\bY} I_E(\Phi_{C_1(\bY),C_2(\bX)}) 
\le &
\frac{e^{\psi(s|W^E,P_{\mix, {\cal X}})}}{L^s s},
\quad 0 < \forall s <1.\Label{4-27-2-c}
\end{align}
}

Next, we consider a special class of channels.
When the channel $W^E$ is additive, i.e., $W^E_x(z)=P(z-x)$,
\red{
(\ref{12-26-2}) implies 
\begin{align}
\rE_{\bX,\bY} I_E(\Phi_{C_1(\bY),C_2(\bX)}) 
\le &
\frac{|{\cal X}|^s e^{-\tilde{H}_{1+s}(X|P)}}{L^s s} 
\end{align}
for $0 < \forall s \le 1$.
In this case,
the equation 
$\psi(s|W^E,P_{\mix, C_1+x})= \psi(s|W^E,P_{\mix, C_1})$ 
holds for any $x$.
Thus, 
(\ref{2-26-1}) and the concavity of $e^{\phi(s|W^E,p)}$ (Lemma \ref{l1})
imply that 
\begin{align}
\psi(s|W^E,P_{\mix, C_1})
\le
\phi(s|W^E,P_{\mix, C_1})
\le \phi(s|W^E,P_{\mix, {\cal X}})
.\Label{2-13-1}
\end{align}
Thus, combining (\ref{4-27-1}), (\ref{2-13-1}), and (\ref{12-26-2-b}), 
we obtain
\begin{align}
\rE_{\bX} I_E(\Phi_{C_1,C_2(\bX)}) 
\le 
\frac{|{\cal X}|^s
e^{- (1-s)\tilde{H}_{1+\frac{s}{1-s}}(X|P)}}{L^s s}  
\Label{4-27-2}
\end{align}
for $0 < \forall s \le 1$.}

Similarly, when the channel $W^E$ is general additive, i.e., 
$W^E_x(z,z')=P^{X,Z'}(z-x,z')$,
we obtain
\red{
\begin{align}
\rE_{\bX} I_E(\Phi_{C_1,C_2(\bX)}) 
\le &
\frac{|{\cal X}|^s e^{- (1-s)\tilde{H}_{1+\frac{s}{1-s}}(X|Z'|P^{X,Z'})}}{L^s s}
\Label{4-27-3} \\
\rE_{\bX\red{,\bY}} I_E(\Phi_{C_1\red{(\bY),}C_2(\bX)}) 
\le &
\frac{|{\cal X}|^s e^{-\tilde{H}_{1+s}(X|Z'|P^{X,Z'})}}{L^s s} 
\end{align}
for $0 < \forall s <1$.}

In the following discussion, we assume that
$\cX$ is an $n$-dimensional vector space $\bF_q^n$
over the finite field $\bF_q$.
Then, the subcode $C_2(\bX)$ of the random linear privacy amplification
can be constructed with small complexity.
That is, when $C_1$ is equivalent to $\bF_q^m$,
an ensemble of the subcodes $C_2(\bX)$
satisfying Condition \ref{C2} can be generated 
from only the $m-1$ independent random variables $X_1, \ldots, X_{m-1}$
on the finite field $\bF_q$ as follows.

When $|C_2(\bX)|=q^k$,
we choose the subcode $C_2(\bX)$ as the kernel of the 
the concatenation of Toeplitz matrix and the identity
$(\bX,I)$ of the size $m \times (m-k)$ 
given in Appendix \ref{stoep}.
Then, 
the encoding $\{Q_{[x]}\}_{[x]\in C_1/C_2(\bX)}$ 
is constructed as follows.
When the sent message is $x \in \bF_q^k$,
it is transformed to 
$(b,x-\bX b)^T \in \bF_q^m$,
where $b=(b_1, \ldots, b_{k})$ 
are $k$ independent random variables.
This process forms the encoding $\{Q_{[x]}\}_{[x]\in C_1/C_2(\bX)}$
because the set $\{(b, -\bX b)^T|b \in  \bF_q^k\}$ is equal to $C_2(\bX)$.
This can be checked using the fact 
that $(\bX,I)(b,x-\bX b)^T = x $ and 
the set $\{(b,-\bX b)^T|b \in  \bF_q^k\}$
forms a $k$-dimensional space.

Therefore, if the error correcting code 
$C_1$ can be constructed with effective encoding and decoding times
and $W^E$ is additive or general additive,
the code $\Phi_{C_1,C_2(\bX)}$ for a wire-tap channel 
satisfying the inequality (\ref{4-27-2}) or (\ref{4-27-3})
can be constructed by using random linear privacy amplification.

Furthermore, for the $n$-fold discrete memoryless case of the wire-tap channel
$W^B, W^E$,
it is possible to achieve the rate
$I(P_{\mix,\cX}:W^B)-I(P_{\mix,\cX}:W^E)$
by a combination of this error correcting 
and random linear privacy amplification
when an error correcting code 
attaining the Shannon rate $I(P_{\mix,\cX}:W^B)$
is available and the channel $W^E$ is general additive, i.e., $W^{E}_{x}(z,z')=P^{X,Z'}(z-x,z')$.
In this case, when the sacrifice information rate is $R$, 
as follows from the discussion of Section \ref{s6} and (\ref{4-27-3}),
the exponent of Eve's information is greater than
$\max_{0 \le s \le 1}\frac{s (R-\log |{\cal X}|)+ \tilde{H}_{1+s}(X|Z'|P^{X,Z'})}{\red{1+s}} 
=
\max_{0 \le s \le 1}
\frac{s}{\red{1+s}} (R-\log |{\cal X}|+ H_{1+s}(X|Z'|P^{X,Z'}) )$.

This method is very useful when the channels 
$W^B$ and $W^E$ are additive.
However, even if the channels are not additive or general additive,
this method is still useful because 
it requires 
only a linear code and random privacy amplification,
which is simpler requirement than 
that of the random coding method given in the proof of Theorem \ref{3-6}
while this method cannot attain the optimal rate.

\section{Secret key agreement}\Label{s5}
Next, following Maurer\cite{Mau93}, we apply the above discussions to secret key agreement,
in which, 
Alice, Bob, and Eve are assumed to have initial random variables 
$a\in {\cal A}$, $b\in {\cal B}$, and $e\in {\cal E}$, respectively.
The task for Alice and Bob is to share a common random variable almost independent of Eve's random variable $e$ by using a public communication.
The quality is evaluated by three quantities:
the size of the final common random variable,
the probability that their final variables coincide, and
the mutual information between Alice's final variables and Eve's random variable.
In order to construct a protocol for this task,
we assume that the set ${\cal A}$ has a module structure 
(any finite set can be regarded as a cyclic group).
Then, the objective of secret key agreement can be realized by applying the code of a wire-tap channel
as follows.
First, 
Alice generates another uniform random variable $x$ and sends
the random variable $x':= x-a$.
Then, 
the distribution of the random variables $b,x'$ ($e,x'$) accessible to Bob (Eve)
can be regarded as the output distribution of the channel $x \mapsto W^{B}_x$
($x \mapsto W^{E}_x$).
The channels $W^{B}$ and $W^{E}$ are given as follows.
\begin{align}
W^{B}_x(b,x')&= P^{AB}(x-x',b) \nonumber \\
W^{E}_x(e,x')&= P^{AE}(x-x',e),\Label{4-26-1}
\end{align}
where $P^{AB}(a,b)$ 
($P^{AE}(a,e)$) 
is the joint probability between
Alice's initial random variable $a$ and
Bob's (Eve's) initial random variable $b$ ($e$).
Hence, the channel $W^E$ is general additive.

Applying Theorem \ref{3-6} to the uniform distribution $P_{\mix}^A$,
for any numbers $M$ and $L$,
there exists a code $\Phi$ such that
\begin{align*}
|\Phi| & =M \\
\epsilon_B(\Phi) & \le 2 \min_{0\le s\le1} (ML)^s |\cA|^{-s} e^{-(1+s)\tilde{H}_{\frac{1}{1+s}}(A|B|P^{A,B})} \\
I_E(\Phi) & \le 2 
\min_{0\le s\le1} \frac{|{\cal A}|^s e^{-\tilde{H}_{1+s}(A|E|P^{A,E})}}{s L^s }
\end{align*}
because 
$e^{\phi(-s|W^B,P_{\mix,\cA})}=|\cA|^{-s} e^{-(1+s)\tilde{H}_{\frac{1}{1+s}}(A|B|P^{A,B})}$.
and 
$\psi(s|W^E,P_{\mix,\cA})=
s \log |{\cal A}|-\tilde{H}_{1+s}(A|E|P^{A,E}) 
=
s (\log |{\cal A}|-H_{1+s}(A|E|P^{A,E}) )$.

In particular, 
when $\cX$ is an $n$-dimensional vector space $\bF_q^n$
over the finite field $\bF_q$ and
the joint distribution between 
$A$ and $B$($E$) 
is the $n$-fold independent and identical distribution (i.i.d.)
of $P^{A,B}$ ($P^{A,E}$), respectively,
the relation
$
\tilde{H}_{1+s}(A^n|E^n|(P^{A,E})^n)
=n \tilde{H}_{1+s}(A|E|P^{A,E})$
holds.
Thus, there exists a code $\Phi_n$ for any integers $L_n,M_n$,
and any probability distribution $p$ on $\cX$ such that
\begin{align}
|\Phi_n| &=M_n \nonumber \\
\epsilon_B(\Phi) & \le 2
\min_{0\le s\le 1}
(M_n L_n)^{s}
|{\cal A}|^{-ns} e^{-n(1+s)\tilde{H}_{\frac{1}{1+s}}(A|B|P^{A,B})}
\nonumber
\\
I_E(\Phi_n) & \le  2
\min_{0\le s\le1} \frac{|{\cal A}|^{ns} e^{-n \tilde{H}_{1+s}(A|E|P^{A,E})}}{s L_n^s }.
\Label{2-13-2}
\end{align}
Hence, 
the achievable rate of this protocol
is equal to 
\begin{align*}
&I(P_{\mix,\cA}:W^B)- I(P_{\mix,\cA}:W^E)\\
=& H(P^B)+H(P_{\mix,\cA})- H(P^{A,B}) \\
&- (H(P^E)+H(P_{\mix,\cA})- H(P^{A,E}))\\
= &H(P^B)+H(P^A)- H(P^{A,B}) \\
&- (H(P^E)+H(P^A)- H(P^{A,E}))\\
= &I(A:B)-I(A:E)
=H(A|E)-H(A|B),
\end{align*}
which was obtained by Maurer\cite{Mau93}
and Ahlswede-Csisz\'{a}r\cite{AC93}.
Here, since the channels $W^B$ and $W^E$ can be regarded as general additive,
we can apply the discussion in Section \ref{s3}.
That is,
the bound (\ref{2-13-2}) can be attained 
with the combination of a linear code and random privacy amplification,
which is given in Section \ref{s3}.

\section{Discussion}
We have derived an upper bound for Eve's information in secret key generation from a common random number without communication
when a universal$_2$ hash function is applied.
Since our bound is based on the R\'{e}nyi entropy of order $1+s$ for $s \in [0,1]$,
it can be regarded as an extension of Bennett et al \cite{BBCM}'s result with the R\'{e}nyi entropy of order 2.

Applying this bound to the wire-tap channel, we obtain an 
upper bound for Eve's information,
which yields an exponential upper bound.
This bound improves on the existing bound \cite{Hayashi}. 
Further, when the error correction code is given by a linear code
and when the channel is additive or general additive,
the privacy amplification is given by the concatenation of Toeplitz matrix and the identity.
Finally, 
our result has been applied to secret key agreement with public communication.

\section*{Acknowledgments}
This research
was partially supported by 
a Grant-in-Aid for Scientific Research in the Priority Area `Deepening and Expansion of Statistical Mechanical Informatics (DEX-SMI)', No. 18079014
and a MEXT Grant-in-Aid for Young Scientists (A) No. 20686026.
The author is grateful to Professor Ryutaroh Matsumoto for a helpful comment for proof of Theorem 2 
and inequalities (\ref{4-27-2}) and (\ref{4-27-3}), 
and interesting discussions.
The author thanks Professors Renato Renner and Shun Watanabe for helpful discussions.
In particular, he greatly thanks Professor Shun Watanabe 
for allowing him to including his example mentioned in Appendix III.
He is also grateful to the
referees for helpful comments concerning this manuscript.

\bibliographystyle{IEEE}

\appendices

\section{Proof of Theorem \ref{thm1}}
\Label{s7}
The concavity of $x \mapsto x^s$ implies that
\begin{align*}
& \rE_{\bX}e^{-\tilde{H}_{1+s}(X|P\circ f_{\bX}^{-1})}
=
\rE_{\bX} 
\sum_{i=1}^M
P\circ f_{\bX}^{-1}(i)
P\circ f_{\bX}^{-1}(i)^s \\
= & \sum_x 
P(x) \rE_{\bX} 
(\sum_{x': f_{\bX}(x)=f_{\bX}(x')} P(x'))^s \\
\le &
\sum_x 
P(x) 
(
\rE_{\bX} 
\sum_{x': f_{\bX}(x)=f_{\bX}(x')} P(x'))^s.
\end{align*}
Condition \ref{C1} guarantees that
\begin{align*}
\rE_{\bX} 
\sum_{x': f_{\bX}(x)=f_{\bX}(x')} P(x')
\le &
P(x) + \sum_{x\neq x'}P(x')
\frac{1}{M} \\
\le &
P(x) + \frac{1}{M}.
\end{align*}
Since any two positive numbers $x$ and $y$ satisfy 
$(x+y)^s \le x^s +y^s$ for $0 \le s \le 1$,
\begin{align*}
(P(x) + \frac{1}{M})^s
\le P(x)^s + \frac{1}{M^s}.
\end{align*}
Hence,
\begin{align*}
& \rE_{\bX}e^{-\tilde{H}_{1+s}(X|P\circ f_{\bX}^{-1})}
\le
\sum_x 
P(x) 
(P(x)^s + \frac{1}{M^s}) \\
= &
\sum_x P(x)^{1+s}
+ \frac{1}{M^s}
=
e^{-\tilde{H}_{1+s}(X|P)}
+ \frac{1}{M^s}.
\end{align*}
Therefore, 
taking the expectation with respect to the random variable $E$,
we have
\begin{align}
\rE_{\bX}e^{-\tilde{H}_{1+s}(A|E|P^{f_{\bX}(A),E})}
\le
e^{-\tilde{H}_{1+s}(A|E|P^{A,E})}
+ \frac{1}{M^s}.\Label{5-14-2}
\end{align}
The concavity of the logarithm implies 
\begin{align*}
\tilde{H}_{1+s}(A|E|P^{A,E})\le s H(A|E).
\end{align*}
Thus, From (\ref{5-14-2}),
the concavity of the logarithm yields
that 
\begin{align*}
& s \rE_{\bX} H(f_{\bX}(A)|E)
\ge
\rE_{\bX} \tilde{H}_{1+s}(A|E|P^{A,E}) \\
\ge &
-\log 
\rE_{\bX} e^{-\tilde{H}_{1+s}(A|E|P^{A,E})} \\
\ge &
-\log (e^{-\tilde{H}_{1+s}(A|E|P^{A,E})}
+ \frac{1}{M^s})\\
= &
s\log M
- \log (
1+ M^s e^{-\tilde{H}_{1+s}(A|E|P^{A,E})}) \\
\ge &
s\log M 
- M^s e^{-\tilde{H}_{1+s}(A|E|P^{A,E})},
\end{align*}
where the last inequality follows from 
the logarithmic inequality $\log (1+x) \le x$.
Therefore, we obtain (\ref{5-14-1}).

\section{Toeplitz matrix}\Label{stoep}
The concatenation of Toeplitz matrix and the identity
$(\bX,I)$ of size $m \times (m-k)$ 
on the finite filed $\bF_q$ is given as follows.
First, we choose an $m-1$ random variables $X_1, \ldots, X_{m-1}$
on the finite filed $\bF_q$.
$I$ is the $(m-k) \times (m-k)$ identity matrix and 
the $k \times (m-k)$ matrix $\bX =(X_{i,j})$ 
is defined by 
the $m-1$ random variables $X_1, \ldots, X_{m-1}$
as follows.
\begin{align*}
X_{i,j}=X_{i+j-1} .
\end{align*}
This matrix is called a Toeplitz matrix.

Now, we prove that 
the $m \times (m-k)$ matrices $(\bX,I)$ satisfy 
Condition \ref{C2}.
More precisely, we show the following.
(1) An element $(x,y)^T \in \bF_q^k \oplus \bF_q^{-(m-k)}$ 
belongs to the kernel of $(\bX,I)$
with probability $q^{k}$ if $x \neq 0$ and $y \neq 0$.
(2) It does not belong to the kernel of the
$m \times (m-k)$ matrix $(\bX,I)$
if $x = 0$ and $y \neq 0$.

Indeed, since (2) is trivial, we will show (1).
For $x=(x_1, \ldots, x_k)$, 
we let $i$ be the minimum index $i$ such that $x_i \neq 0$.
We fix the $k-i$ random variables $X_{i+(m-k)-1}, \ldots,X_{m-1}$.
That is, we show that
the element $(x,y)^T$ belongs to the kernel with probability $q^{k}$
when the $k-i$ random variables $X_{i+(m-k)-1}, \ldots,X_{m-1}$ are fixed.
Then, the condition $\bX x+y=0$ can be expressed as the following
$m-k$ conditions.
\begin{align*}
X_i x_1 &= - \sum_{j=i+1}^{k} X_{j} x_j -y_1 \\
X_{i+1} x_2 &= - \sum_{j=i+1}^{k} X_{j+1} x_j -y_2 \\
& \vdots \\
X_{i+m-k-2} x_{m-k-1} &= - \sum_{j=i+1}^{k} X_{j+m-k-2} x_j -y_{m-k-1} \\
X_{i+m-k-1} x_{m-k} &= - \sum_{j=i+1}^{k} X_{j+m-k-1} x_j -y_{m-k} .
\end{align*}
The $(m-k)$-th condition does not depend on the
$m-k-1$ variables $X_i, \ldots X_{i+(m-k)-1}$.
Hence, this condition only depends on the variable $X_{i+m-k-1}$.
Therefore, the $(m-k)$-th condition holds with probability
$1/q$.
Similarly, we can show that
the $(m-k-1)$-th condition holds with probability
$1/q$ under the $(m-k)$-th condition.
Thus, the $(m-k)$-th condition and the $(m-k-1)$-th condition 
hold with probability $1/q^2$.
Repeating this discussion inductively,
we can conclude that
all $m-k$ conditions hold with probability $q^{-(m-k)}$.

\section{Two leaked information criteria}\Label{a3}
In this appendix, we explain an example, in which,
the leaked information criterion based on the variational distance is small 
but  
the leaked information criterion based on the mutual information is large.
This example is proposed by Shun Watanabe\cite{Wata}.
The former criterion is given as \cite{Cannetti}
\begin{align*}
d_1( P^{A,E} ,P^A_{\mix} \times P^E),
\end{align*}
where $P^A_{\mix}$ is the uniform distribution on ${\cal A}$
and the variational distance is given as $d_1(P,Q):=\sum_{x}|P(x)-Q(x)|$.
Pinsker inequality \cite{CKbook} guarantees that
\begin{align*}
& d_1( P^{A,E} ,P^A_{\mix} \times P^E) \\
\le &
d_1( P^{A,E} ,P^A \times P^E)
+
d_1( P^A \times P^E ,P^A_{\mix} \times P^E) \\
\le &
D( P^{A,E} \| P^A \times P^E)^2
+d_1( P^A ,P^A_{\mix}) \\
= & I(A:E)^2+d_1( P^A ,P^A_{\mix}),
\end{align*}
where $D(P\|Q):= \sum_{x}P(x) (\log P(x)-\log Q(x))$.
This inequality shows that 
when $d_1( P^A ,P^A_{\mix})$ and $I(A:E)^2$ are close to zero,
$d_1( P^{A,E} ,P^A_{\mix} \times P^E)$ is also close to zero.

Assume that 
the Eve's distribution $P^E$ is the uniform distribution,
and ${\cal E}={\cal A}$.
For any small real number $\epsilon >0$,
we define a subset ${\cal S} \subset {\cal E}$ such that
$P^E({\cal S})=1-\epsilon$.
The conditional distribution $P^{A|E}$ is assumed to be given as
\begin{align*}
P^{A|E}(a|e):=
\left\{
\begin{array}{ll}
\frac{1}{|{\cal E}|} & \hbox{ if } e \in {\cal S} \\
\delta_{a,e} & \hbox{ if }  \in {\cal S}^c ,
\end{array}
\right.
\end{align*}
where $\delta_{a,e}$ is $1$ when $a=e$, and is $0$ otherwise.
Then, the leaked information criterion based on the variational distance is 
evaluated as
\begin{align*}
&d_1( P^{A,E} ,P^A_{\mix} \times P^E)
=\sum_{e\in {\cal E}} P^E(e) d_1( P^{A|E} ,P^A_{\mix})\\
=&\sum_{e\in {\cal S}} P^E(e) d_1( P^{A|E} ,P^A_{\mix})
+\sum_{e\in {\cal S}^c} P^E(e) d_1( P^{A|E} ,P^A_{\mix})\\
\le & 2 \epsilon.
\end{align*}
In oder to evaluate the leaked information criterion based on the mutual information,
we focus on the probability
\begin{align*}
P_e:= P^{A,E}\{a\neq e \}.
\end{align*}
Fano inequality\cite{CKbook} yields that
\begin{align*}
H(E|A)\le 1 + P_e \log |{\cal E}|.
\end{align*}
Since $P_e \le 1- \epsilon$,
\begin{align*}
& I(A:E)=H(E)-H(E|A)
\ge H(E)-  1 - P_e \log |{\cal E}| \\
=&  \log |{\cal E}|-  1 - P_e \log |{\cal E}|
\ge -  1 + \epsilon \log |{\cal E}|.
\end{align*}
In particular, when ${\cal E}=\{0,1\}^{n^2}$ and $\epsilon =\frac{1}{n}$,
\begin{align*}
d_1( P^{A,E} ,P^A_{\mix} \times P^E)
\le \frac{2}{n} ,\quad
I(A:E)
 \ge n-  1 .
\end{align*}

This example shows that 
even if 
$d_1( P^{A,E} ,P^A_{\mix} \times P^E)$
is close to zero,
there is a possibility that
$I(A:E)$ is not close to zero.
Hence, we cannot guarantee 
the security based on mutual information
from the security based on variational distance
while we can guarantee 
the security based on variational distance from the security based on mutual information
when $d_1( P^{A} ,P^A_{\mix})$ is close to zero.
Therefore, 
the leaked information criterion based on the mutual information 
is more restrictive than that based on variational distance.


\begin{figure*}[!t]

\section{Proof of (\ref{2-25-1})}
\Label{a4}

Since 
\begin{align}
I(p,W)=\sum_x p(x) D(W^E_x\|W^E_p) \le \sum_x p(x) D(W^E_x\|Q) 
\Label{eq6}
\end{align}
holds for any distribution $Q$,
\begin{align}
& \rE_{\bY} E_{\bX|\bY} I_E(\Phi(\bX,\bY)') 
\le 
\rE_{\bY} E_{\bX|\bY} 
\frac{1}{LM}
\sum_{k=1}^{LM}
D(\frac{1}{L} \sum_{k':f_{\bX}(k')=f_{\bX}(k) } W^E_{f_{\Phi(\bY)}(k')}
\| W^E_p) \Label{eq1} \\
= &
\rE_{\bY} E_{\bX|\bY} 
\frac{1}{LM}
\sum_{k=1}^{LM}
\sum_{y} 
\frac{1}{L}\sum_{k'':f_{\bX}(k'')=f_{\bX}(k) } W^E_{f_{\Phi(\bY)}(k'')}(y)
(\log (\frac{1}{L}\sum_{k':f_{\bX}(k')=f_{\bX}(k) } W^E_{f_{\Phi(\bY)}(k')}(y))-
\log W^E_p(y)) \nonumber \\
= &
\rE_{\bY} E_{\bX|\bY} 
\frac{1}{LM}
\sum_{k=1}^{LM}
\sum_{y} 
W^E_{f_{\Phi(\bY)}(k)}(y)
(\log (\frac{1}{L} \sum_{k':f_{\bX}(k')=f_{\bX}(k) } W^E_{f_{\Phi(\bY)}(k')}(y))-
\log W^E_p(y)) \nonumber \\
\le &
\rE_{\bY} 
\frac{1}{LM}
\sum_{k=1}^{LM}
\sum_{y} 
W^E_{f_{\Phi(\bY)}(k)}(y)
( \log ( \frac{1}{L} W^E_{f_{\Phi(\bY)}(k)}(y)+E_{\bX|\bY} \frac{1}{L} \sum_{k'\neq k:f_{\bX}(k')=f_{\bX}(k) } W^E_{f_{\Phi(\bY)}(k')}(y))-
\log W^E_p(y)) \Label{eq2} \\
\le &
\rE_{\bY} 
\frac{1}{LM}
\sum_{k=1}^{LM}
\sum_{y} 
W^E_{f_{\Phi(\bY)}(k)}(y)
( \log (\frac{1}{L} W^E_{f_{\Phi(\bY)}(k)}(y)+\frac{1}{M L} \sum_{k'\neq k} W^E_{f_{\Phi(\bY)}(k')}(y))-
\log W^E_p(y)) \Label{eq3} \\
= &
\frac{1}{LM}
\sum_{k=1}^{LM}
\sum_{y} 
\rE_{\bY_k} 
W^E_{f_{\Phi(\bY)}(k)}(y)
\rE_{\bY|\bY_{k}}
(\log (\frac{1}{L} W^E_{f_{\Phi(\bY)}(k)}(y)+\frac{1}{M L} \sum_{k'\neq k} W^E_{f_{\Phi(\bY)}(k')}(y))-
\log W^E_p(y)) \nonumber  \\
\le &
\frac{1}{LM}
\sum_{k=1}^{LM}
\sum_{y} 
\rE_{\bY_k} 
W^E_{f_{\Phi(\bY)}(k)}(y)
(\log (\frac{1}{L} W^E_{f_{\Phi(\bY)}(k)}(y)+\frac{1}{M L} \rE_{\bY|\bY_{k}}
\sum_{k'\neq k} W^E_{f_{\Phi(\bY)}(k')}(y))-
\log W^E_p(y)) \Label{eq4} \\
\le &
\frac{1}{LM}
\sum_{k=1}^{LM}
\sum_{y} 
\rE_{\bY_k} 
W^E_{f_{\Phi(\bY)}(k)}(y)
(\log (\frac{1}{L} W^E_{f_{\Phi(\bY)}(k)}(y)+W^E_p(y)) -\log W^E_p(y)) \Label{eq5}\\
= &
\frac{1}{LM}
\sum_{k=1}^{LM}
\sum_{y} 
\rE_{\bY_k} 
W^E_{f_{\Phi(\bY)}(k)}(y)
\log ( 1+ \frac{1}{L}\frac{W^E_{x}(y)}{W^E_p(y)})\nonumber \\
= &
\frac{1}{LM}
\sum_{k=1}^{LM}
\sum_{y} 
\sum_{x\in {\cal X}}
p(x)
W^E_x(y)
\log ( 1+ \frac{1}{L}\frac{W^E_{x}(y)}{W^E_p(y)})
= 
\sum_{y} 
\sum_{x\in {\cal X}} p(x)
W^E_x(y)
\log ( 1+ \frac{1}{L}\frac{W^E_{x}(y)}{W^E_p(y)}),
\nonumber 
\end{align}
where the random variable $f_{\Phi(\bY)}(k)$ is simplified to $\bY_k$.
In the above derivation, 
(\ref{eq1}) follows from (\ref{eq6}),
(\ref{eq2}) and (\ref{eq4}) follow from the concavity of $\log x$,
and
(\ref{eq3}) and (\ref{eq5}) follow from Conditions \ref{C1} and \ref{C12}.

Since the inequalities $(1+x)^s \le 1+ x^s$ 
and $\log (1+x) \le x$
hold for any positive $x$ and $0 < s \le 1$,
the inequalities
\begin{align}
\log (1+x) \le \frac{\log (1+x)^s}{s}\le
\frac{\log (1+x^s)}{s}\le 
\frac{x^s}{s}\Label{4-27-14}
\end{align}
hold.
Using this inequality,
we obtain
\begin{align}
\sum_{y} 
\sum_{x\in {\cal X}} p(x)
W^E_x(y)
\log ( 1+ \frac{1}{L}\frac{W^E_{x}(y)}{W^E_p(y)})
\le 
\sum_{y} 
\sum_{x\in {\cal X}} p(x)
W^E_x(y)
\frac{1}{s L^s}\frac{W^E_{x}(y)^s}{W^E_p(y)^s})
=
\frac{1}{s L^s} e^{\psi(s|W^E,p)}
\Label{4-27-8-1},
\end{align}
which implies (\ref{2-25-1}).

\end{figure*}

\newpage

\begin{figure*}[!t]
\section{Proof of (\ref{4-27-2-c})}
\Label{a5}
Since
\begin{align*}
& I_E(\Phi_{C_1,C_2(\bX)}) 
=
\sum_y 
\frac{1}{|C_1|}\sum_{x'\in C_1}W^E_{x'}(y)
(\log ( \frac{1}{|C_2(\bX )|} \sum_{x'': x'-x''\in C_2(\bX )} W^E_{x''}(y)) 
- \log ( \frac{1}{|C_1|} \sum_{x''' \in C_1} W^E_{x''}(y)) )\\
\le &
\sum_y 
\frac{1}{|C_1|}\sum_{x'\in C_1}W^E_{x'}(y)
(\log ( \frac{1}{|C_2(\bX )|} \sum_{x'': x'-x''\in C_2(\bX )} W^E_{x''}(y)) 
- \log (  W^E_{P_{\mix, {\cal X}}}(y)) ),
\end{align*}
we have
\begin{align}
& \rE_{\bY} E_{\bX|\bY} I_E(\Phi_{C_1(\bY),C_2(\bX)}) \nonumber \\
\le &
\rE_{\bY} E_{\bX|\bY}
\sum_y 
\frac{1}{ML} \sum_{x'\in C_1(\bY)}W^E_{x'}(y)
(\log ( \frac{1}{L} \sum_{x'': x'-x''\in C_2(\bX )} W^E_{x''}(y)) 
- \log (  W^E_{P_{\mix, {\cal X}}}(y)) )\nonumber \\
= &
\rE_{\bY} E_{\bX|\bY}
\sum_y 
\frac{1}{ML}\sum_{x'\in C_1(\bY)}W^E_{x'}(y)
(\log ( \frac{1}{L} W^E_{x'}(y)+
\frac{1}{L} \sum_{x'': x'-x''\in C_2(\bX ), x'\neq x''} W^E_{x''}(y)
) 
- \log (  W^E_{P_{\mix, {\cal X}}}(y)) ) \nonumber \\
\le &
\rE_{\bY} 
\sum_y 
\frac{1}{ML}\sum_{x'\in C_1(\bY)}W^E_{x'}(y)
(\log ( \frac{1}{L} W^E_{x'}(y)+
E_{\bX|\bY} \frac{1}{L} \sum_{x'': x'-x''\in C_2(\bX ), x'\neq x''} W^E_{x''}(y)
) 
- \log (  W^E_{P_{\mix, {\cal X}}}(y)) ) \Label{eq7} \\
\le &
\rE_{\bY} 
\sum_y 
\frac{1}{ML} \sum_{x'\in C_1(\bY)}W^E_{x'}(y)
(\log ( \frac{1}{L} W^E_{x'}(y)+
\frac{1}{L}\frac{L}{ML} \sum_{x''\in C_1 \setminus\{x'\}} W^E_{x''}(y)) 
- \log (  W^E_{P_{\mix, {\cal X}}}(y)) ) \Label{eq8} \\
= &
\sum_y 
\frac{1}{|{\cal X}|} \sum_{x'\in {\cal X}}W^E_{x'}(y)
\rE_{\bY|x'\in C_1(\bY)} 
(\log ( \frac{1}{L} W^E_{x'}(y)+ \frac{1}{ML} \sum_{x''\in C_1(\bY) \setminus\{x'\}} W^E_{x''}(y) )
- \log (  W^E_{P_{\mix, {\cal X}}}(y)) ) \nonumber \\
\le &
\sum_y 
\frac{1}{|{\cal X}|} \sum_{x'\in {\cal X}}W^E_{x'}(y)
(\log ( \frac{1}{L} W^E_{x'}(y)+ \frac{1}{ML} \rE_{\bY|x'\in C_1(\bY)} \sum_{x''\in C_1(\bY) \setminus \{x'\}} W^E_{x''}(y) )
- \log (  W^E_{P_{\mix, {\cal X}}}(y)) ) \Label{eq9} \\
\le &
\sum_y 
\frac{1}{|{\cal X}|} \sum_{x'\in {\cal X}}W^E_{x'}(y)
(\log ( \frac{1}{L} W^E_{x'}(y)+ \frac{1}{|{\cal X}|} \sum_{x''\in {\cal X} \setminus \{x'\}} W^E_{x''}(y) )
- \log (  W^E_{P_{\mix, {\cal X}}}(y)) ) \Label{eq10} \\
\le &
\sum_y 
\frac{1}{|{\cal X}|} \sum_{x'\in {\cal X}}W^E_{x'}(y)
(\log ( \frac{1}{L} W^E_{x'}(y)+ W^E_{P_{\mix,{\cal X}}}(y) )
- \log (  W^E_{P_{\mix, {\cal X}}}(y)) ) \nonumber \\
= &
\sum_y 
\frac{1}{|{\cal X}|} \sum_{x'\in {\cal X}}W^E_{x'}(y)
\log ( 1+ \frac{1}{L} \frac{W^E_{x'}(y)}{W^E_{P_{\mix,{\cal X}}}(y)} ),
\nonumber
\end{align}
where $\rE_{\bY|C}$ is the conditional expectation concerning the random variable $\bX$ when the condition $C$ holds.
In the above derivation,
(\ref{eq7}) and (\ref{eq9}) follow from the concavity of $\log x$,
and
(\ref{eq8}) and (\ref{eq10}) follow from Conditions \ref{C2} and \ref{C4}.

Using (\ref{4-27-14}),
we obtain (\ref{4-27-2-c}).
\end{figure*}

\end{document}